\documentclass{article}
\usepackage{nameref}
\usepackage[pdfstartview=FitH,pdfpagemode=UseNone,pagebackref,colorlinks,linkcolor=blue,filecolor = blue,
citecolor = blue, urlcolor = blue]{hyperref}

\usepackage{amsmath, amssymb, amsthm}
\usepackage{geometry}
\usepackage{algorithm}
\usepackage{algorithmic}
\usepackage{tikz}
\usetikzlibrary{arrows.meta,positioning}

\title{Total Variation Distance Estimation through Domain Reduction}
\author{Arnab Bhattacharyya\\{\small University of Warwick} \and Graham Cormode\\{\small University of Oxford} \and Yucheng Fu\\{\small University of Hong Kong} \and Kuldeep S.~Meel\\{\small Georgia Institute of Technology}}
\date{}

\newtheorem{definition}{Definition}[section]
\newtheorem{lemma}{Lemma}[section]
\newtheorem{theorem}{Theorem}[section]

\newcommand{\ignore}[1]{}
\begin{document}

\maketitle
\begin{abstract}
Computing the total variation (TV) distance between succinctly represented high-dimensional distributions is generally intractable. We give an FPRAS for TV distance between mixtures of product distributions and, more generally, for a natural class of structured probabilistic circuits.

Our main technique is a novel application of domain reduction: Given a family of feature vectors indexed by assignments, we use Lewis-weight sampling to replace the assignment domain by a polynomial-size weighted subset that simultaneously approximates the sum of absolute values of every linear projection. For mixtures of product distributions, we construct such reduced domains incrementally over the coordinates, obtaining the first FPRAS with running time polynomial in both the dimension and the number of mixture components. We then extend the approach to smooth, structured-decomposable probabilistic circuits with a common structured architecture.
\end{abstract}

\section{Introduction}

High-dimensional probability distributions arise throughout theoretical
computer science, machine learning, statistics, and statistical physics.
Such distributions are typically represented succinctly: although a
distribution on $[q]^n$ has $q^n$ possible outcomes, familiar classes such
as product distributions and graphical models can
be specified using only polynomially many parameters. A basic algorithmic
problem is to compare two distributions given through succinct
representations, without explicitly enumerating their exponentially large
domains.

A canonical measure of discrepancy is the \emph{total variation distance}
\[
    d_{\mathrm{TV}}(P,Q)
    =
    \frac{1}{2}\sum_x |P(x)-Q(x)|.
\]
Total variation has several equivalent statistical interpretations: for
instance, it is the maximum discrepancy $|P(A)-Q(A)|$ over events $A$, and
the minimum disagreement probability over couplings of $P$ and $Q$.
Computationally, however, it behaves quite differently from distances such
as KL divergence, Hellinger distance, and $\chi^2$-divergence: even for
product distributions, TV distance does not decompose coordinatewise.

This distinction was made concrete by Bhattacharyya et
al.~\cite{ijcai2023p387,bhattacharyya2025total}, who showed that exactly computing the TV distance
between two product distributions is $\#\mathsf{P}$-complete. This
initiated a recent line of work on relative approximation of TV distance
for succinctly represented distributions. Feng, Guo, Jerrum, and
Wang~\cite{feng2023simple} gave a polynomial-time randomized relative-approximation
algorithm for arbitrary product distributions. Feng, Liu, and
Liu~\cite{feng2024soda} subsequently gave a deterministic FPTAS, and extended
their approach to Markov chains. Beyond product distributions,
Bhattacharyya et al.~\cite{bhattacharyya2024inference} reduced relative TV approximation for
same-structure Bayesian networks to probabilistic inference, yielding
FPRASs whenever the corresponding inference problem is tractable, while
Feng, Liu, and Yang~\cite{pmlr-v291-feng25a} obtained relative-approximation
algorithms for several families of spin systems.

A natural next case is a \emph{mixture of product distributions}. Let
\[
    P=\sum_{i=1}^{k_1}\alpha_i P_i,
    \qquad
    Q=\sum_{j=1}^{k_2}\beta_j Q_j,
\]
where every $P_i$ and $Q_j$ is a product distribution on $[q]^n$.
Mixtures retain a compact description, but the latent mixture component
introduces global dependence among all $n$ coordinates. 
They have a long history of study in learning theory (e.g., \cite{feldman2008learning,jain2014learning,gordon2021source,gordon2024identification}).
\cite{bhattacharyya2025computational} 
gave a simple polynomial-time algorithm for checking equivalence ($d_{\mathrm{TV}}=0$) 
between two mixtures of product distributions.
Recently, Feng,
Fu, Yang, and Zhang~\cite{feng2026computing} gave an FPRAS whose running time has
exponential dependence on the total number of mixture components.
Independent recent work of Fu~\cite{Fu26} gives a deterministic
FPTAS with a similar exponential dependence on
$k=k_1+k_2$. Thus, these algorithms run in polynomial time when $k$ is
constant, but not when the number of mixture components is part of the
input. Indeed, the two recent works~\cite{feng2026computing, Fu26} explicitly leave open the
following question:
\begin{quote}
\centering
\emph{Can the TV distance between mixtures of product distributions be
approximated in time polynomial simultaneously in $n$, $q$, the number of
mixture components $k$, and $1/\varepsilon$?}
\end{quote}

Looking beyond this, we can view the mixture of product distributions as a simple circuit whose leaves are the input variables, so that the middle layer computes the product distributions and the top layer computes their sum. 
The gates are then $\oplus$ and $\otimes$, which perform the sum and product operations on distributions, respectively. 
Allowing other wiring graphs defines a more general way to describe distributions compactly as a \textit{probabilistic circuit}.  
This family of probabilistic circuits has important subcases, such as smoothness and decomposability (defined formally below). 
Further generalizations allow other types of gate to perform (bilinear) operations on their inputs. 
Since probabilistic circuits naturally express product distributions as a special case, they inherit the hardness bounds. 
Probabilistic circuits have been advocated as a powerful unifying representation for inference tasks in machine learning and statistics \cite{ChoiVergariVanDenBroeck20,ZhangWangArenasVanDenBroeck25}.
However, there have been no prior results for approximation of TV distance over probabilistic circuits. 

\paragraph{Our results.}
We answer the above open question affirmatively. 
Our main result gives a relative
approximation to the TV distance between two mixtures of product
distributions in time polynomial in the dimension, the alphabet size, the
number of mixture components, and the desired relative accuracy.

\begin{theorem}[Informal]
Let $P$ and $Q$ be mixtures of $k_1$ and $k_2$ product distributions,
respectively, over $[q]^n$, and let $k=k_1+k_2$. There is an FPRAS which,
given $\varepsilon,\delta>0$, outputs $\widetilde d$ satisfying
\[
    (1-\varepsilon)d_{\mathrm{TV}}(P,Q)
    \leq
    \widetilde d
    \leq
    (1+\varepsilon)d_{\mathrm{TV}}(P,Q)
\]
with probability at least $1-\delta$, using
\[
    \operatorname{poly}
    \left(
        n,q,k,\frac{1}{\varepsilon},
        \log\frac{1}{\delta}
    \right)
\]
arithmetic operations.
\end{theorem}

In particular, the dependence on the number of mixture components is
polynomial rather than exponential. 
Moreover, the running time is
independent of the numerical value of $d_{\mathrm{TV}}(P,Q)$; in the
real-arithmetic model, the algorithm is strongly polynomial in the natural
input parameters.
This represents a significant breakthrough for the complexity of $d_{\mathrm{TV}}$, showing a FPRAS that is truly polynomial in the size of the compressed input. 

Our algorithm is based on a general technique that we refer to as 
\emph{domain reduction} for distributions. 
Rather than approximating the probability of
individual outcomes for the exponential sized domain, we repeatedly compress an exponentially growing set
of partial assignments to a polynomial-size weighted subset that
approximately preserves every relevant linear functional. 
The reduction is obtained through $\ell_1$ subspace sparsification using Lewis weights.
The fact that all linear functionals are preserved simultaneously is what
allows a reduced domain constructed at one stage to remain valid under
arbitrary future coordinates.

In fact, our full result extends much more broadly than mixtures of product distributions.  
The same principle of domain reduction for distributions extends substantially beyond mixtures. 
We show that domain reductions compose through the sum and product operations of
probabilistic circuits, and more generally through any bilinear operator. 
As a consequence, we obtain a
relative-approximation algorithm for two smooth, structured-decomposable
probabilistic circuits respecting a common v-tree; apart from the common
decomposition of variable scopes, the two circuits may have unrelated
gates and wiring. 
Thus mixtures of product distributions arise as a
particularly simple, sequential instance of a more general compositional phenomenon.

Weighted tree automata (WTAs) provide another natural application of
our approach. These classical models assign weights to labeled trees
and include probabilistic and latent-variable context-free grammars
as special cases \cite{FulopVogler09,RabusseauBalleCohen16}.
We consider nonnegative WTAs on a fixed binary tree, whose normalized
weights define distributions over leaf labelings. The weight of each
labeling can be computed bottom-up, but TV distance involves summing
absolute differences over potentially exponentially many labelings.
The bilinear transitions of these models allow us to apply the same
domain-reduction principle. We obtain an FPRAS for the TV distance
between the distributions induced by two such automata on the same
fixed binary tree.

\paragraph{Formalization}: We formalized our result for mixtures as well as smooth structured decomposable probabilsitic circuits in Lean4. For our Lean formalization, we also formalized the sparsification using lewis weights. In particular, the entire formalization is end to end and only depends on the classical axioms of Lean. The end-to-end formalization is available at \url{https://github.com/meelgroup/mixtureslean}

\subsection{Technical overview}
\label{sec:technical-overview}

We first explain the main idea for two product distributions
$P,Q$ over $\Omega^n$.  Write
$P(x)=\prod_{t=1}^n P_t(x_t)$ and
$Q(x)=\prod_{t=1}^n Q_t(x_t)$.
For each coordinate $t$ and value $x_t\in\Omega$, define the local
feature vector
$r_t(x_t)=[P_t(x_t),Q_t(x_t)]^\top$.
For a prefix $x_{\leq t}=(x_1,\ldots,x_t)$, let
$R_{\leq t}(x_{\leq t})
=\bigotimes_{j=1}^t r_j(x_j)$,
where $\otimes$ denotes component-wise multiplication.  Thus the two
coordinates of $R_{\leq t}(x_{\leq t})$ are the probabilities of the
prefix under $P$ and $Q$, respectively.  At the end, writing
$R(x)=R_{\leq n}(x)$ and $w=(1,-1)^\top$, we have
\[
    2d_{\mathrm{TV}}(P,Q)
    =
    \sum_{x\in\Omega^n}
    |w^\top R(x)|.
\]

The main difficulty to overcome is that the number of prefixes grows exponentially with
$t$.  
Instead of storing all of them, after iteration $t$ our approach is to maintain a small weighted set
$C_t=\{(z,\mu,v)\}$, where $z$ is a retained prefix,
$v=R_{\leq t}(z)$, and $\mu>0$ is its weight.  For a $\mu$-weighted set $C$
and a query vector $y$, write
$E(C,y)=\sum_{(z,\mu,v)\in C}\mu|y^\top v|$.

The key point is to understand what information about the prefix domain
must be preserved.  For a future assignment
$x_{>t}=(x_{t+1},\ldots,x_n)$, define
$S_{>t}(x_{>t})
=\bigotimes_{j=t+1}^n r_j(x_j)$.
Then
\[
    P(x)-Q(x)
    =
    \bigl(w\otimes S_{>t}(x_{>t})\bigr)^\top
    R_{\leq t}(x_{\leq t}).
\]
Thus every possible suffix induces a linear query
$y=w\otimes S_{>t}(x_{>t})$ on the current prefix features.

Crucially, at iteration $t$ the future multiplier
$S_{>t}(x_{>t})$ has not yet been computed: the remaining coordinates have not yet
been incorporated into the coreset.  
We therefore cannot tailor the
compression to a particular future query.  
However, we can leverage the fact that whatever $S_{>t}(x_{>t})$ turns out to be, it is \textit{some} linear query. 
As a result, we ask for the stronger guarantee
\[
    E(C_t,y)
    \approx
    \sum_{x_{\leq t}\in\Omega^t}
    |y^\top R_{\leq t}(x_{\leq t})|
    \qquad
    \text{simultaneously for \textit{every} }y\in\mathbb{R}^2.
\]
This uniform guarantee therefore automatically covers every query that may arise from any choice of the future coordinates.

In the terminology of dimensionality reduction, 
such a compression is precisely an $\ell_1$ subspace embedding.  Given
$C_{t-1}$, we first extend every retained prefix by every
$x_t\in\Omega$.  If $(z,\mu,v)\in C_{t-1}$, the extension $(z,x_t)$
has feature vector $v\otimes r_t(x_t)$ and inherits weight $\mu$.
Let $U_t$ denote the resulting candidate set.  If $A_t$ is the matrix
whose rows are the weighted feature vectors of $U_t$, then
$E(U_t,y)=\|A_ty\|_1$.  Sampling rows according to their $\ell_1$
Lewis weights produces a much smaller weighted set $C_t$ satisfying
\[
    (1-\delta)E(U_t,y)
    \leq
    E(C_t,y)
    \leq
    (1+\delta)E(U_t,y)
    \qquad
    \text{for all }y.
\]
The size of $C_t$ is polynomial in $1/\delta$ and the feature
dimension, and is independent of the exponentially large number of
prefixes represented by it.

\paragraph{Mixtures of product distributions.}
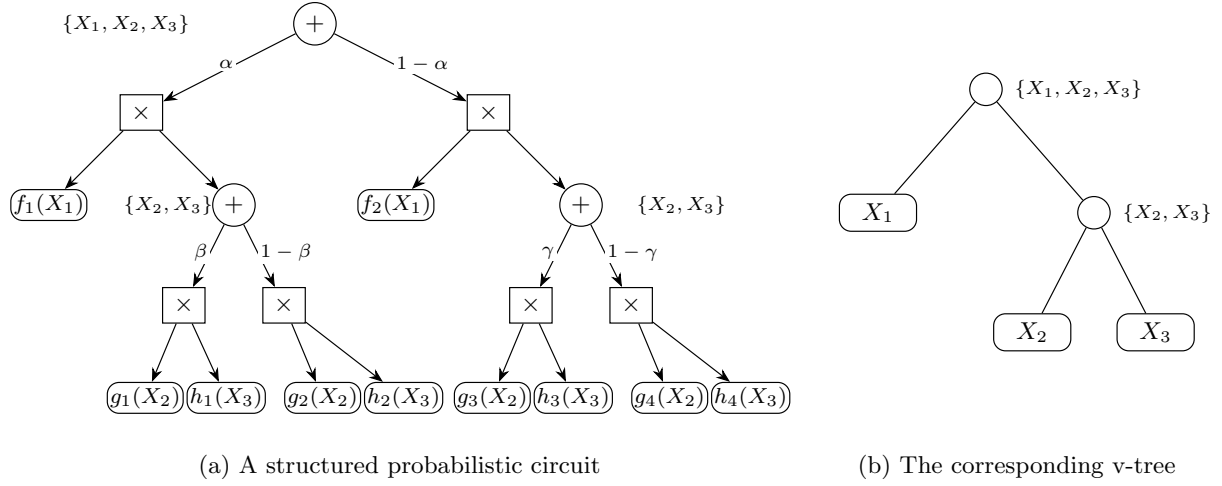
\begin{figure}[t]
\centering
\resizebox{0.97\linewidth}{!}{%
\begin{tikzpicture}[
    >=Stealth,
    font=\small,
    scale=0.92,
    transform shape,
    sum/.style={
        circle, draw,
        inner sep=1pt,
        minimum size=5.5mm
    },
    prod/.style={
        rectangle, draw,
        inner sep=1pt,
        minimum width=5.5mm,
        minimum height=4.6mm
    },
    leaf/.style={
        rounded corners, draw,
        inner sep=1.5pt,
        minimum width=10mm,
        minimum height=4.8mm,
        align=center
    },
    vtnode/.style={
        circle, draw,
        inner sep=0pt,
        minimum size=4.2mm
    },
    edgelabel/.style={
        font=\scriptsize,
        fill=white,
        inner sep=0.8pt
    },
    scopelabel/.style={
        font=\scriptsize
    },
    cleaf/.style={
    rounded corners,
    draw,
    inner xsep=1pt,
    inner ysep=0.5pt,
    minimum width=0pt,
    minimum height=4mm,
    font=\footnotesize,
    align=center
}
]

% =========================================================
% LEFT PANEL: probabilistic circuit
% =========================================================

% Root mixture
\node[sum] (root) at (0,4.8) {$+$};

% First product layer
\node[prod] (p1) at (-2.25,3.65) {$\times$};
\node[prod] (p2) at ( 2.25,3.65) {$\times$};

\draw[->] (root) -- node[edgelabel,pos=.45,left] {$\alpha$} (p1);
\draw[->] (root) -- node[edgelabel,pos=.45,right] {$1-\alpha$} (p2);

% X1 leaves and {X2,X3} mixture nodes
\node[cleaf] (x1a) at (-3.45,2.45) {$f_1(X_1)$};
\node[sum]  (s1)  at (-1.05,2.45) {$+$};

\node[cleaf] (x1b) at ( 1.05,2.45) {$f_2(X_1)$};
\node[sum]  (s2)  at ( 3.45,2.45) {$+$};

\draw[->] (p1) -- (x1a);
\draw[->] (p1) -- (s1);

\draw[->] (p2) -- (x1b);
\draw[->] (p2) -- (s2);

% Lower product gates
\node[prod] (q11) at (-1.70,1.15) {$\times$};
\node[prod] (q12) at (-0.40,1.15) {$\times$};

\node[prod] (q21) at ( 2.80,1.15) {$\times$};
\node[prod] (q22) at ( 4.10,1.15) {$\times$};

\draw[->] (s1) -- node[edgelabel,pos=.45,left] {$\beta$} (q11);
\draw[->] (s1) -- node[edgelabel,pos=.45,right] {$1-\beta$} (q12);

\draw[->] (s2) -- node[edgelabel,pos=.45,left] {$\gamma$} (q21);
\draw[->] (s2) -- node[edgelabel,pos=.45,right] {$1-\gamma$} (q22);

% Bottom leaves
\node[cleaf] (a21) at (-2.20,-0.05) {$g_1(X_2)$};
\node[cleaf] (a31) at (-1.15,-0.05) {$h_1(X_3)$};

\node[cleaf] (a22) at ( 0.10,-0.05) {$g_2(X_2)$};
\node[cleaf] (a32) at ( 1.15,-0.05) {$h_2(X_3)$};

\node[cleaf] (b21) at ( 2.30,-0.05) {$g_3(X_2)$};
\node[cleaf] (b31) at ( 3.35,-0.05) {$h_3(X_3)$};

\node[cleaf] (b22) at ( 4.60,-0.05) {$g_4(X_2)$};
\node[cleaf] (b32) at ( 5.65,-0.05) {$h_4(X_3)$};

\draw[->] (q11) -- (a21);
\draw[->] (q11) -- (a31);

\draw[->] (q12) -- (a22);
\draw[->] (q12) -- (a32);

\draw[->] (q21) -- (b21);
\draw[->] (q21) -- (b31);

\draw[->] (q22) -- (b22);
\draw[->] (q22) -- (b32);

% Optional scope annotations
\node[scopelabel,anchor=east] at (-1.5,4.8) {$\{X_1,X_2,X_3\}$};
\node[scopelabel,anchor=east] at (-1.2,2.45) {$\{X_2,X_3\}$};
\node[scopelabel,anchor=west] at (4.05,2.45) {$\{X_2,X_3\}$};

\node at (1.1,-0.95) {(a) A structured probabilistic circuit};

% =========================================================
% RIGHT PANEL: v-tree
% =========================================================

\node[vtnode] (vr) at (8.7,3.95) {};
\node[scopelabel,anchor=west] at (8.95,3.95) {$\{X_1,X_2,X_3\}$};

\node[leaf] (vx1) at (7.30,2.35) {$X_1$};
\node[vtnode] (v23) at (10.10,2.35) {};
\node[scopelabel,anchor=west] at (10.35,2.35) {$\{X_2,X_3\}$};

\node[leaf] (vx2) at (9.30,0.80) {$X_2$};
\node[leaf] (vx3) at (10.90,0.80) {$X_3$};

\draw (vr) -- (vx1);
\draw (vr) -- (v23);
\draw (v23) -- (vx2);
\draw (v23) -- (vx3);

\node at (9.1,-0.95) {(b) The corresponding v-tree};

\end{tikzpicture}%
}
\caption{
A structured probabilistic circuit and a v-tree that it respects.
Sum gates compute weighted mixtures, while product gates multiply
functions on disjoint variable sets. The root of the v-tree represents
the partition $\{X_1,X_2,X_3\}=\{X_1\}\sqcup\{X_2,X_3\}$, and the lower
internal node represents $\{X_2,X_3\}=\{X_2\}\sqcup\{X_3\}$. Every
product gate in the circuit follows one of these two partitions.
}
\label{fig:pc-vtree}
\end{figure}

The same idea extends directly to mixtures.  Suppose
$P=\sum_{i=1}^{k_1}\alpha_iP_i$ and
$Q=\sum_{i=1}^{k_2}\beta_iQ_i$, where every component is a product
distribution on $\Omega^n$, and let $k=k_1+k_2$. We now take
\[
    r_t(x_t)
    =
    [P_{1,t}(x_t),\ldots,P_{k_1,t}(x_t),
      Q_{1,t}(x_t),\ldots,Q_{k_2,t}(x_t)]^\top
\]
and
$w=[\alpha_1,\ldots,\alpha_{k_1},-\beta_1,\ldots,-\beta_{k_2}]^\top$.
The accumulated feature vector is again
$R_{\leq t}(x_{\leq t})
=\bigotimes_{j=1}^t r_j(x_j)$, and
$w^\top R(x)=P(x)-Q(x)$.

The argument above is unchanged, except that the query vectors now lie
in $\mathbb{R}^{k}$.
In particular, at iteration $t$ the future
multiplier $S_{>t}$ is still as yet unknown, so the coreset must preserve
$E(C_t,y)$ simultaneously for every $y\in\mathbb{R}^{k}$.
Since $\ell_1$ Lewis-weight sampling produces a coreset whose size is
polynomial in the feature dimension, the number of retained prefixes is
polynomial in $k$.  This is the step that removes the exponential
dependence on the number of mixture components in previous approaches.

\paragraph{Probabilistic circuits.}
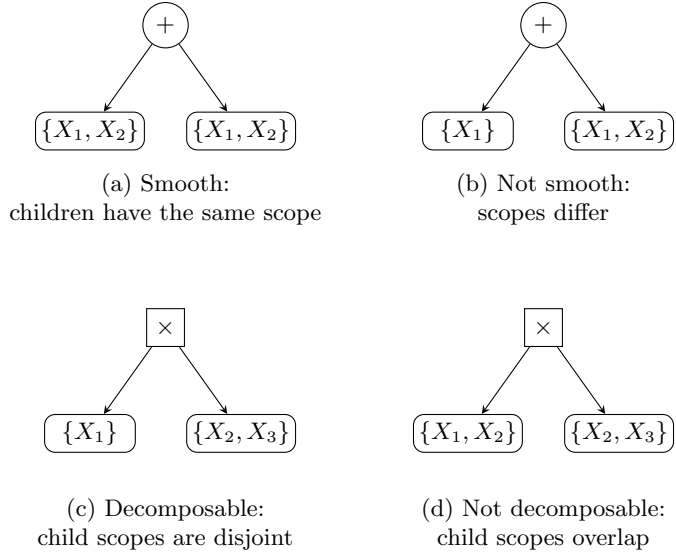
\begin{figure}[t]
\centering
\begin{tikzpicture}[
    >=stealth,
    sum/.style={circle, draw, inner sep=1pt, minimum size=6mm},
    prod/.style={rectangle, draw, inner sep=1pt, minimum size=5mm},
    leaf/.style={rounded corners, draw, inner sep=2pt, minimum width=12mm, minimum height=5mm},
    good/.style={draw=black},
    bad/.style={draw=black},
    font=\small
]

% ------------------------------------------------
% (a) smooth sum gate
% ------------------------------------------------
\node[sum] (sa) at (0,1.8) {$+$};
\node[leaf] (sa1) at (-1,0.4) {$\{X_1,X_2\}$};
\node[leaf] (sa2) at (1,0.4) {$\{X_1,X_2\}$};
\draw[->] (sa) -- (sa1);
\draw[->] (sa) -- (sa2);
\node[align=center] at (0,-0.5) {(a) Smooth:\\ children have the same scope};

% ------------------------------------------------
% (b) non-smooth sum gate
% ------------------------------------------------
\node[sum] (sb) at (5,1.8) {$+$};
\node[leaf] (sb1) at (4,0.4) {$\{X_1\}$};
\node[leaf] (sb2) at (6,0.4) {$\{X_1,X_2\}$};
\draw[->] (sb) -- (sb1);
\draw[->] (sb) -- (sb2);
\node[align=center] at (5,-0.5) {(b) Not smooth:\\ scopes differ};

% ------------------------------------------------
% (c) decomposable product gate
% ------------------------------------------------
\node[prod] (pc) at (0,-2.2) {$\times$};
\node[leaf] (pc1) at (-1,-3.6) {$\{X_1\}$};
\node[leaf] (pc2) at (1,-3.6) {$\{X_2,X_3\}$};
\draw[->] (pc) -- (pc1);
\draw[->] (pc) -- (pc2);
\node[align=center] at (0,-4.8) {(c) Decomposable:\\ child scopes are disjoint};

% ------------------------------------------------
% (d) non-decomposable product gate
% ------------------------------------------------
\node[prod] (pd) at (5,-2.2) {$\times$};
\node[leaf] (pd1) at (4,-3.6) {$\{X_1,X_2\}$};
\node[leaf] (pd2) at (6,-3.6) {$\{X_2,X_3\}$};
\draw[->] (pd) -- (pd1);
\draw[->] (pd) -- (pd2);
\node[align=center] at (5,-4.8) {(d) Not decomposable:\\ child scopes overlap};

\end{tikzpicture}
\caption{
The two key structural conditions used in our circuit algorithm.
Smoothness concerns sum gates; decomposability concerns product gates.
}
\label{fig:smooth-decomp}
\end{figure}

The same idea extends beyond mixtures of products to a broader class of
structured probabilistic models.  A \emph{probabilistic circuit} is a
directed acyclic graph whose leaves are univariate distributions, whose
sum gates compute weighted sums of their children, and whose product
gates compute products of their children~\cite{ChoiVergariVanDenBroeck20,ZhangWangArenasVanDenBroeck25}.  
Thus sum gates represent
mixtures, while product gates represent factorizations over disjoint
variable sets.  We will consider two circuits over the same set of
variables $X_1,\ldots,X_n$.

To describe the structure of such a circuit, it is convenient to use a
\emph{v-tree}: a full binary tree whose leaves are in bijection with the
variables.  Each node of the v-tree corresponds to the set of variables
in the leaves of its subtree.  A circuit is \emph{structured} with
respect to a v-tree if every product gate decomposes its scope according
to the left/right partition induced by some internal v-tree node.  In
particular, if a v-tree node corresponds to a variable set
$S=S_h\sqcup S_\ell$, then every product gate associated with $S$
splits into one child over $S_h$ and one child over $S_\ell$. See Figure \ref{fig:pc-vtree}.

At each region $S$ of the common v-tree, we define a feature vector
$\Phi_S(x_S)$ that records the values of all gates of the two circuits
having scope $S$.  
As before, we seek a small weighted subset of
assignments to $S$ that preserves the $\ell_1$ norm of every linear
query on these feature vectors.

The structural conditions needed to achieve this are the standard ones from the probabilistic-circuit literature.  
\emph{Smoothness} means that the
children of each sum gate have the same scope, so a sum gate acts as a
linear transformation of the feature vector for that scope.  Hence sum
gates do not require any new domain reduction.  \emph{Decomposability}
means that the children of each product gate have disjoint scopes.
Therefore, at a product region $S=S_h\sqcup S_\ell$, the value of a
parent gate is a bilinear expression in the child feature vectors
$\Phi_{S_h}(x_h)$ and $\Phi_{S_\ell}(x_\ell)$. See Figure \ref{fig:smooth-decomp}.

This bilinear structure is exactly what allows the inductive step.  A
linear query on the parent features becomes a bilinear form in the two
child feature vectors; fixing either child turns it into a linear query
on the other.  Thus the two child coresets can be applied successively,
after which we sparsify the resulting Cartesian-product domain.  In this
way, the one-dimensional sequence of coordinates for mixtures of product
distributions is replaced by a bottom-up traversal of the common v-tree.

\subsection{Related work}
\label{sec:related-work}

\paragraph{TV distance for succinctly represented distributions.}
At the level of unrestricted succinct representations, TV-distance
approximation is already connected to central complexity classes.
Sahai and Vadhan~\cite{SV03} showed that \emph{Statistical Difference}---given two circuits sampling distributions, distinguish the case in which
their TV distance is small from the case in which it is  large---is complete
for $\mathsf{SZK}$. Earlier work of Goldreich, Sahai, and
Vadhan~\cite{GoldreichSV99} established closely related completeness results for
non-interactive statistical zero knowledge. These results indicate that
efficient TV approximation cannot be expected for arbitrary circuit
descriptions, motivating the study of structured representations for
which the problem becomes tractable. 

\ignore{For product distributions, Bhattacharyya et al.~\cite{BGMMPV23} proved
that exact TV computation is $\#\mathsf{P}$-complete, while also giving
approximation algorithms for restricted cases. Feng, Guo, Jerrum, and
Wang~\cite{FGJW23} subsequently gave an FPRAS for arbitrary product
distributions. Their algorithm uses a coupling-based importance-sampling
argument and remains polynomial-time even when the TV distance is very
small. Feng, Liu, and Liu~\cite{FLL24} gave a deterministic FPTAS based on
sparsification of the likelihood-ratio distribution, and extended their
techniques to Markov chains. Bhattacharyya et
al.~\cite{BGMMPV24} established a structure-preserving reduction from
relative TV approximation to probabilistic inference, yielding FPRASs for
same-structure Bayesian networks whenever inference can be performed
efficiently. Feng, Liu, and Yang~\cite{FLY25} obtained
relative-approximation algorithms for several classes of spin systems.}

In the introduction, we identified the recent sequence of works
(including \cite{bhattacharyya2025total, feng2023simple, feng2024soda})
on approximating the TV distance efficiently.
Most closely related to the present work are recent algorithms for
mixtures of product distributions. Feng, Fu, Yang, and
Zhang~\cite{feng2026computing, Fu26} give a relative-approximation algorithm
whose running time is polynomial for every fixed number of mixture
components, but exponential in the total number $k$ of components.
They explicitly identify polynomial dependence
on $k$ as an open problem. Our result resolves this question.

\paragraph{Distribution testing and additive approximation.}
There is also a large literature on testing whether two unknown
distributions are identical or close given sample access. Non-tolerant
closeness testing asks to distinguish
\[
    P=Q
    \qquad\text{from}\qquad
    d_{\mathrm{TV}}(P,Q)>\varepsilon,
\]
and has been extensively studied both for arbitrary distributions and
for structured classes such as product distributions
\cite{canonne2020testing,daskalakis2017square}. In \emph{tolerant} closeness testing, the goal is
instead to distinguish
\[
    d_{\mathrm{TV}}(P,Q)\leq \varepsilon_1
    \qquad\text{from}\qquad
    d_{\mathrm{TV}}(P,Q)\geq \varepsilon_2.
\]
This problem is closely related to additive approximation of TV distance:
an additive estimator immediately gives a tolerant tester, while
tolerant testers at a sequence of thresholds can be used to obtain an
additive estimate.

Tolerant identity and closeness testing for product distributions were
studied systematically by Bhattacharyya et al.~\cite{bhattacharyya2020efficient},
who showed how learning
algorithms can be converted into efficient additive TV-distance
estimators for several classes of structured high-dimensional
distributions, including Bayesian networks, Ising models, and Gaussian
distributions. For unrestricted distributions over a domain of size
$N$, the work of Valiant and Valiant~\cite{valiant2011estimating,valiant2017automatic} and subsequent
work gives nearly optimal estimators for TV distance and related
symmetric distributional functionals.

The present work concerns a different access model: the distributions are
given through their succinct parameters, and our goal is a
\emph{relative}, rather than additive, approximation. Relative
approximation is particularly demanding when $d_{\mathrm{TV}}(P,Q)$ is
very small: a generic additive estimator does not distinguish a distance
of $2^{-\Theta(n)}$ from zero. Our domain-reduction approach instead uses
the algebraic structure of the succinct representation to preserve such
small distances multiplicatively. 

\paragraph{Probabilistic circuits.}
Probabilistic circuits (PCs), also known in important special cases as
sum-product networks, are a tractable model class built from univariate
distributions using weighted sum and product gates.  Under structural
conditions such as smoothness and decomposability, they support exact
marginal inference in time linear in the circuit size; see, e.g.,
\cite{PoonDomingos11,PeharzTschiatschekPernkopf15,ChoiVergariVanDenBroeck20}.
A particularly relevant line of work studies which \emph{operations} on
circuits preserve tractability.  Vergari et al.~\cite{VergariEtAl21}
develop a compositional framework for tractable circuit operations and,
among other consequences, obtain polynomial-time algorithms for several
information-theoretic quantities, including KL divergence, under suitable
structural compatibility assumptions on the two circuits.  In their
framework, computing $D_{\mathrm{KL}}(P\|Q)$ is tractable when the two
circuits can be combined through the required product/quotient/logarithm
operations while remaining in a tractable PC class.  More recently,
Zhang et al.~\cite{ZhangWangArenasVanDenBroeck25} studied
\emph{restructuring} structured PCs, showing how to transform a circuit
to respect a target v-tree and thereby enabling further tractable
operations such as multiplication across different structured
decompositions.  Our result is different in flavor: rather than reducing
TV approximation to a tractable closure property of the circuit class, we
develop a domain-reduction technique that directly yields a relative
approximation algorithm for TV distance on pairs of smooth,
structured-decomposable circuits sharing a common v-tree.

\paragraph{$\ell_1$ row sampling and sparsification.}
Given a matrix
$A\in\mathbb{R}^{N\times d}$, an $\ell_1$ subspace embedding seeks a
reweighted subset of the rows such that
\[
    (1-\varepsilon)\|Ax\|_1
    \leq
    \|\widetilde A x\|_1
    \leq
    (1+\varepsilon)\|Ax\|_1
    \qquad\text{for every }x\in\mathbb{R}^d.
\]
This question has a long history in convex geometry.
Talagrand~\cite{Tal90} showed that
$O(d\varepsilon^{-2}\log d)$ appropriately reweighted rows suffice.
Cohen and Peng~\cite{cohen2015lp} gave an efficient randomized construction
based on $\ell_p$ Lewis weights, which may be viewed as the analogue of
statistical leverage scores for $\ell_p$ norms.  In particular, for
$p=1$, sampling and reweighting
$O(d\varepsilon^{-2}\log d)$ rows according to their Lewis weights
preserves $\|Ax\|_1$ simultaneously for every $x$ with high probability.
This machinery, together with efficient algorithms for approximating
Lewis weights, is the sparsification primitive used in our domain
reduction.

Very recently, Reis and Rothvoss~\cite{RR26} resolved the
existential sparsification problem optimally up to constants: every
$A\in\mathbb{R}^{N\times d}$ admits a reweighting supported on only
$O(d/\varepsilon^2)$ rows that preserves every $\ell_1$ norm
$\|Ax\|_1$ to within a $(1\pm\varepsilon)$ factor. Their construction for general
$\ell_1$ sparsification is not polynomial time, however; so it does not
directly replace the efficient Lewis-weight sampling needed in our
algorithm.
\section{Preliminaries}

\subsection{Problem Setup}
Let $P$ and $Q$ be mixture distributions over a discrete product space $\Omega^n = \prod_{t=1}^n \Omega_t$. We assume $P$ and $Q$ are mixtures of $k_1$ and $k_2$ product distributions, respectively, and write $k=k_1+k_2$ for the total number of components:
\begin{equation}
    P(x) = \sum_{i=1}^{k_1} \alpha_i P_i(x), \quad Q(x) = \sum_{i=1}^{k_2} \beta_i Q_i(x)
\end{equation}
where $\alpha, \beta$ are the mixture weights, and $P_i(x) = \prod_{t=1}^n P_{i,t}(x_t)$, $Q_i(x) = \prod_{t=1}^n Q_{i,t}(x_t)$. For any coordinate $t \in [n]$ and domain element $x_t \in \Omega_t$, we define the $k$-dimensional local feature vector $\mathbf{r}_t(x_t)$:
\begin{equation}
    \mathbf{r}_t(x_t) = \left[ P_{1,t}(x_t), \dots, P_{k_1,t}(x_t), \ Q_{1,t}(x_t), \dots, Q_{k_2,t}(x_t) \right]^T
\end{equation}
Let $\otimes$ denote component-wise multiplication (the Hadamard product). For any prefix assignment $x_{\le t} = (x_1, \dots, x_t) \in \prod_{j=1}^t \Omega_j$, the accumulated unnormalized feature vector is:
\begin{equation}
    \mathbf{R}_{\le t}(x_{\le t}) = \bigotimes_{j=1}^t \mathbf{r}_j(x_j)
\end{equation}
When $t=n$, we write $\mathbf{R}(x) := \mathbf{R}_{\le n}(x)$. We define the constant $k$-dimensional weight vector $\mathbf{w}$ containing the mixture coefficients:
\begin{equation}
    \mathbf{w} = \left[ \alpha_1, \dots, \alpha_{k_1}, \ -\beta_1, \dots, -\beta_{k_2} \right]^T
\end{equation}

By linearity, the inner product $\mathbf{w}^T \mathbf{R}(x)$ yields exactly $P(x) - Q(x)$. Thus, the Total Variation distance is the sum of the absolute linear projections over the entire domain:
\begin{equation}
    d_{\mathrm{TV}}(P, Q) = \frac{1}{2} \sum_{x \in \Omega^n} \left| P(x) - Q(x) \right| = \frac{1}{2} \sum_{x \in \Omega^n} \left| \mathbf{w}^T \mathbf{R}(x) \right|
\end{equation}

\subsection{\texorpdfstring{$L_1$}\ ~Subspace Embeddings and Domain Coresets}

To prevent the state space from growing exponentially as we iterate over the dimensions, we replace the exponentially large exact domain at step $t$ with a small weighted subset of assignments that preserves all relevant linear tests.

Let $\mathcal{C} = \{(z_i, \mu_i, \mathbf{v}_i)\}_{i=1}^N$ be a weighted set of domain assignments $z_i$ with associated scalar weights $\mu_i > 0$ and feature vectors $\mathbf{v}_i \in \mathbb{R}^d$. Let $A \in \mathbb{R}^{N \times d}$ be the matrix whose $i$-th row is $\mu_i \mathbf{v}_i^T$. For any query vector $\mathbf{u} \in \mathbb{R}^d$, the weighted sum of absolute projections is the $L_1$ norm of $A\mathbf{u}$:
\begin{equation}
    \sum_{i=1}^N \mu_i |\mathbf{u}^T \mathbf{v}_i| = \|A\mathbf{u}\|_1
\end{equation}

We rely on the following foundational theorem for dimensionality reduction in $L_1$ spaces \cite{cohen2015lp}.
\begin{theorem}[$L_1$ Subspace Embedding via Lewis Weights \cite{cohen2015lp}]
\label{thm:lewis_weights}
Given a matrix $A \in \mathbb{R}^{N \times d}$, a precision parameter $\delta \in (0, 1/2)$, and a failure probability $\eta \in (0, 1)$, there exists a randomized algorithm that outputs a non-negative diagonal sampling matrix $Y \in \mathbb{R}^{N \times N}$ with at most $m = O\left( \frac{d \log d}{\delta^2} \log\frac{1}{\eta} \right)$ non-zero entries. With probability at least $1 - \eta$, for all $\mathbf{u} \in \mathbb{R}^d$ simultaneously:
\begin{equation}
    (1-\delta) \|A\mathbf{u}\|_1 \le \|YA\mathbf{u}\|_1 \le (1+\delta) \|A\mathbf{u}\|_1
\end{equation}
Furthermore, in the arithmetic model, the algorithm computes the sampling matrix $Y$ using $\widetilde{O}(Nd + d^\omega)$ arithmetic operations (\cite[Theorem 5.3.1]{lee2016faster}; \cite[Lemma 2.5]{jambulapati2022improved}), where $\omega \leq 2.4$ is the matrix multiplication exponent.
\end{theorem}
Applying the sampling matrix $Y$ to $A$ selects a sparse subset of at most $m$ assignments with updated weights $\tilde{\mu}_i = Y_{ii} \mu_i$, certifying a uniform relative error of $(1 \pm \delta)$ across all possible linear projections.

\section{Algorithm and Analysis}

We first formalize the algorithm for the case of mixtures of product distributions. At each step $t$, we extend the surviving prefix assignments by one coordinate and sparsify the resulting domain coreset.

\subsection{The FPRAS Algorithm}

Let $\textsc{Sparsify}(\mathcal{U}, m, \eta')$ be a subroutine that takes a candidate assignment coreset $\mathcal{U}$, forms its feature matrix, computes the $L_1$ Lewis weights, and returns a subsampled coreset $\mathcal{C}$ of size at most $m$, succeeding with probability at least $1 - \eta'$.

\begin{algorithm}[H]
\caption{FPRAS for Mixture TV via Assignment Coresets}
\label{alg:fptas}
\begin{algorithmic}[1]
\REQUIRE Marginals for $P, Q$ over $\prod_{t=1}^n \Omega_t$, relative accuracy $\epsilon \in (0, 1)$, failure probability $\eta > 0$.
\STATE Set step-wise error $\delta = \frac{\epsilon}{3n}$ and step-wise failure rate $\eta' = \frac{\eta}{n}$.
\STATE Set sample size $m = \Theta \left( \frac{k \log k}{\delta^2} \log\left(\frac{1}{\eta'}\right) \right)$.
\STATE Initialize root coreset $\mathcal{C}_0 = \{ (\emptyset, 1, \mathbf{1}_{k}) \}$.
\FOR{$t = 1$ \TO $n$}
    \STATE Initialize candidate extension domain $\mathcal{U}_t = \emptyset$.
    \FOR{each assignment state $(z, \mu, \mathbf{v}) \in \mathcal{C}_{t-1}$}
        \FOR{each domain element $x_t \in \Omega_t$}
            \STATE Form extended prefix assignment: $u = (z, x_t)$.
            \STATE Compute feature vector: $\mathbf{u} = \mathbf{v} \otimes \mathbf{r}_t(x_t)$.
            \STATE Add $(u, \mu, \mathbf{u})$ to $\mathcal{U}_t$.  \hfill $\triangleright$ Weight $\mu$ is directly inherited
        \ENDFOR
    \ENDFOR
    \STATE $\mathcal{C}_t \leftarrow \textsc{Sparsify}(\mathcal{U}_t, m, \eta')$.
\ENDFOR
\RETURN $\tilde{D} = \frac{1}{2} \sum_{(z, \mu, \mathbf{v}) \in \mathcal{C}_n} \mu \left| \mathbf{w}^T \mathbf{v} \right|$.
\end{algorithmic}
\end{algorithm}

\subsection{Analysis}

\begin{theorem}[Correctness and Complexity]
\label{thm:main_fptas}
Given $\epsilon, \eta \in (0, 1)$, Algorithm \ref{alg:fptas} outputs a value $\tilde{D}$ such that:
\begin{equation}
    (1 - \epsilon) d_{\mathrm{TV}}(P, Q) \le \tilde{D} \le (1 + \epsilon) d_{\mathrm{TV}}(P, Q)
\end{equation}
with probability at least $1 - \eta$. The algorithm runs in time polynomial in $n, k, 1/\epsilon, \log(1/\eta)$, and the maximum domain size $\max_t |\Omega_t|$.
\end{theorem}

We define the evaluation functional for any assignment coreset $\mathcal{C} = \{(z_i, \mu_i, \mathbf{v}_i)\}$ and any query vector $\mathbf{y} \in \mathbb{R}^{k}$:
\begin{equation}\label{eq:evaluation}
    E(\mathcal{C}, \mathbf{y}) := \sum_{(z, \mu, \mathbf{v}) \in \mathcal{C}} \mu \left| \mathbf{y}^T \mathbf{v} \right|
\end{equation}

\begin{lemma}[One-Step Coreset Guarantee]
\label{lem:one_step}
With probability at least $1 - \eta'$, the $\textsc{Sparsify}$ subroutine succeeds at step $t$. Conditioned on this success, for the candidate extension $\mathcal{U}_t$ and the reduced coreset $\mathcal{C}_t$, it holds simultaneously for all $\mathbf{y} \in \mathbb{R}^{k}$ that:
\begin{equation}
    (1 - \delta) E(\mathcal{U}_t, \mathbf{y}) \le E(\mathcal{C}_t, \mathbf{y}) \le (1 + \delta) E(\mathcal{U}_t, \mathbf{y})
\end{equation}
\end{lemma}
\begin{proof}
Let $A$ be the matrix whose rows are $\mu \mathbf{v}^T$ for $(u, \mu, \mathbf{v}) \in \mathcal{U}_t$. By definition, $E(\mathcal{U}_t, \mathbf{y}) = \|A\mathbf{y}\|_1$. The $\textsc{Sparsify}$ routine applies Theorem \ref{thm:lewis_weights} with failure parameter $\eta'$, producing a reweighted subset matrix $\tilde{A}$ corresponding to the surviving assignments in $\mathcal{C}_t$. With probability at least $1 - \eta'$, the subspace embedding guarantee holds, ensuring $\|\tilde{A}\mathbf{y}\|_1 \in (1 \pm \delta) \|A\mathbf{y}\|_1$ for all $\mathbf{y}$, which corresponds exactly to $E(\mathcal{C}_t, \mathbf{y})$.
\end{proof}

To analyze global propagation, let $\Omega_{>t} = \prod_{j=t+1}^n \Omega_j$ denote the future domain. For any future sequence $x_{>t} \in \Omega_{>t}$, we define the exact future multiplier vector: $\mathbf{S}_{>t}(x_{>t}) = \bigotimes_{j=t+1}^n \mathbf{r}_j(x_j)$.

\begin{definition}[Hybrid Functional]
Let the accumulated TV distance evaluated from the domain coreset $\mathcal{C}_t$ at step $t$ be:
\begin{equation}
    F_t := \frac{1}{2} \sum_{x_{>t} \in \Omega_{>t}} E\Big(\mathcal{C}_t, \ \mathbf{w} \otimes \mathbf{S}_{>t}(x_{>t})\Big)
\end{equation}
\end{definition}
Note that $F_0 = d_{\mathrm{TV}}(P, Q)$ and $F_n = \tilde{D}$.

\begin{lemma}[Propagation of Error]
\label{lem:propagation}
Conditioned on $\textsc{Sparsify}$ succeeding at step $t$, the accumulated TV distance satisfies:
\begin{equation}
    (1 - \delta) F_{t-1} \le F_t \le (1 + \delta) F_{t-1}
\end{equation}
\end{lemma}
\begin{proof}
First, we evaluate the hybrid functional on the exact extension domain $\mathcal{U}_t$ prior to sparsification. By definition of $\mathcal{U}_t$ and the distributivity of the Hadamard product:
\begin{align}
    \frac{1}{2} \sum_{x_{>t}} E\Big(\mathcal{U}_t, \ \mathbf{w} \otimes \mathbf{S}_{>t}(x_{>t})\Big) 
    &= \frac{1}{2} \sum_{x_{>t}} \sum_{(z, \mu, \mathbf{v}) \in \mathcal{C}_{t-1}} \sum_{x_t \in \Omega_t} \mu \left| \big(\mathbf{w} \otimes \mathbf{S}_{>t}(x_{>t})\big)^T (\mathbf{v} \otimes \mathbf{r}_t(x_t)) \right| \nonumber \\
    &= \frac{1}{2} \sum_{x_{\ge t} \in \Omega_{\ge t}} \sum_{(z, \mu, \mathbf{v}) \in \mathcal{C}_{t-1}} \mu \left| \big(\mathbf{w} \otimes \mathbf{S}_{>t-1}(x_{>t-1})\big)^T \mathbf{v} \right| \nonumber \\
    &= \frac{1}{2} \sum_{x_{>t-1}} E\Big(\mathcal{C}_{t-1}, \ \mathbf{w} \otimes \mathbf{S}_{>t-1}(x_{>t-1})\Big) \nonumber \\
    &= F_{t-1}
\end{align}
Second, we bound $F_t$. For any fixed $x_{>t}$, define the query vector $\mathbf{y} = \mathbf{w} \otimes \mathbf{S}_{>t}(x_{>t})$. By Lemma \ref{lem:one_step}, the sparsified coreset $\mathcal{C}_t$ satisfies $(1 - \delta) E(\mathcal{U}_t, \mathbf{y}) \le E(\mathcal{C}_t, \mathbf{y}) \le (1 + \delta) E(\mathcal{U}_t, \mathbf{y})$. Because the sum over $\Omega_{>t}$ is a strictly positive linear operation, summing these inequalities over all $x_{>t}$ yields $(1-\delta) F_{t-1} \le F_t \le (1+\delta) F_{t-1}$.
\end{proof}

\begin{proof}[Proof of Theorem \ref{thm:main_fptas} (Correctness)]
By the union bound, the algorithm succeeds across all $n$ steps with probability at least $1 - n \eta' = 1 - \eta$. Conditioned on success, we telescope Lemma \ref{lem:propagation}:
\begin{equation}
    (1 - \delta)^n F_0 \le F_n \le (1 + \delta)^n F_0
\end{equation}
Substituting $\delta = \epsilon / (3n)$, we have $(1 + \epsilon/3n)^n \le e^{\epsilon/3} \le 1 + \epsilon$, and $(1 - \epsilon/3n)^n \ge 1 - \epsilon$. Since $F_0 = d_{\mathrm{TV}}(P, Q)$ and $F_n = \tilde{D}$, the relative approximation bound holds.
\end{proof}

\begin{proof}[Proof of Theorem \ref{thm:main_fptas} (Complexity)]
At step $t$, let $m_t = |\mathcal{C}_{t-1}|$. The candidate extension domain $\mathcal{U}_t$ contains $N = m_t |\Omega_t|$ assignments with feature dimension $d = k$. Substituting $m = O\left( \frac{k \log k}{\delta^2} \log\frac{1}{\eta'} \right)$ and $\delta = \frac{\epsilon}{3n}$, and using the runtime bound stated in Theorem \ref{thm:lewis_weights}, the number of arithmetic operations per step is bounded by:
\begin{align}
    \widetilde{O} \left( \left(|\Omega_t| \cdot \frac{k \log k}{\delta^2} \log\frac{1}{\eta'}\right) k + k^\omega \right) 
    &= \widetilde{O} \left( |\Omega_t| \frac{n^2 k^2}{\epsilon^2} \log\left(\frac{n}{\eta}\right) + k^\omega \right)
\end{align}
Summing over the $n$ steps, the total number of arithmetic operations is bounded by a polynomial in $n$, $k$, $1/\epsilon$, $\log(1/\eta)$, and $\max_t |\Omega_t|$. This establishes the algorithm as an FPRAS in the arithmetic model.
\end{proof}

\section{Extension to Probabilistic Circuits}
\label{sec:pc}

We generalize the coreset state representation to estimate the Total Variation distance between two distributions $P$ and $Q$ computed by smooth, decomposable probabilistic circuits $C_P$ and $C_Q$ over finite-valued variables $X_1, \dots, X_n$, where $X_i\in\Omega_i$. Both circuits are structured with respect to the same v-tree, but their gates and wiring may otherwise differ. We allow arbitrary nonnegative univariate input functions given explicitly as lookup tables and arbitrary nonnegative sum-gate weights; internal subcircuits need not be normalized, while the two root outputs are required to be the normalized probability mass functions $P$ and $Q$.

\subsection{Circuit Setup and Domain Coresets}
For $R\in\{P,Q\}$ and every gate $g\in C_R$, let $R_g(x_{\mathrm{scope}(g)})$ denote the nonnegative subcircuit function computed at $g$. We identify each region $S$ of the common v-tree with its set of variable indices and write $\Omega_S:=\prod_{i\in S}\Omega_i$. Let $G_S^R$ be the set of gates in $C_R$ whose scope is exactly $S$. We define the region feature vector $\Phi_S(x_S) \in \mathbb{R}^{d_S}$ by concatenating the gate evaluations from the two circuits:
\begin{equation*}
    \Phi_S(x_S) := \left[ P_g(x_S) : g \in G_S^P, \ Q_g(x_S) : g \in G_S^Q \right]^T,
\end{equation*}
where $d_S = |G_S^P|+|G_S^Q|$. Let $W = \max_{R\in\{P,Q\}}\max_S |G_S^R|$ denote the maximum region width of either circuit, so $d_S\leq 2W$.

For an error parameter $\rho\in(0,1)$, a $\rho$-domain reduction for region $S$ is a weighted subset of assignments $\mathcal{C}_S = \{(z_j, \mu_j, \Phi_S(z_j))\}_{j=1}^{m_S}$ such that, for every query vector $a \in \mathbb{R}^{d_S}$,
\begin{equation*}
    (1-\rho)\sum_{x_S \in \Omega_S} \left| \langle a, \Phi_S(x_S) \rangle \right|
    \leq \sum_{j=1}^{m_S} \mu_j \left| \langle a, \Phi_S(z_j) \rangle \right|
    \leq (1+\rho)\sum_{x_S \in \Omega_S} \left| \langle a, \Phi_S(x_S) \rangle \right|.
\end{equation*}
Let $r_P$ and $r_Q$ be the output gates of $C_P$ and $C_Q$, respectively. The root feature vector $\Phi_{\mathrm{root}}(x)$ contains $P_{r_P}(x)$ and $Q_{r_Q}(x)$. Let $a_{\mathrm{TV}} \in \mathbb{R}^{d_{\mathrm{root}}}$ select the former coordinate with coefficient $1$ and the latter with coefficient $-1$. Then
\begin{equation*}
    d_{\mathrm{TV}}(P, Q) = \frac{1}{2} \sum_{x \in \Omega^n} |P_{r_P}(x) - Q_{r_Q}(x)| = \frac{1}{2} \sum_{x \in \Omega^n} \left| \langle a_{\mathrm{TV}}, \Phi_{\mathrm{root}}(x) \rangle \right|
\end{equation*}
Therefore the root coreset yields the estimator
\begin{equation*}
     \widetilde D:=\frac{1}{2} \sum_{j=1}^{m_{\mathrm{root}}} \mu_j \left| \langle a_{\mathrm{TV}}, \Phi_{\mathrm{root}}(z_j) \rangle \right|.
\end{equation*}
\subsection{Bottom-Up Coreset Construction}
Let $q:=\max_{i\in[n]}|\Omega_i|$ be the maximum leaf-domain size, and call a region of the common v-tree active if either circuit has a product gate at that region. At each active product region, all such gates are processed together and one joint sparsification step is performed. Let $L$ be the number of active product regions. If $L=0$, $\textsc{Sparsify}$ is never called and the TV-distance can be computed deterministically and exactly. For $L>0$, we set the local error $\delta = \frac{\epsilon}{3L}$ and failure probability $\eta' = \frac{\eta}{L}$. We construct $\mathcal{C}_S$ bottom-up along the common v-tree:

\paragraph{Leaf Gates.} For a leaf variable $X_i$ with domain $\Omega_i$, we enumerate the exact weighted assignment set $\mathcal{C}_{\{i\}} = \{(a, 1, \Phi_{\{i\}}(a)) : a \in \Omega_i\}$.

\paragraph{Sum Gates.} For $R\in\{P,Q\}$, let $g\in G_S^R$ be a sum gate with inputs $h_1,\dots,h_t$. Smoothness ensures $\mathrm{scope}(h_1)=\dots=\mathrm{scope}(h_t)=S$, and $R_g(x_S)=\sum_{j=1}^t\alpha_{g,j}^R R_{h_j}(x_S)$. Collecting the sum gates of both circuits therefore gives a block-diagonal linear transformation $L_S \in \mathbb{R}^{d_S \times d_S^{\mathrm{old}}}$ of the existing features:
\begin{equation*}
    \Phi_S^{\mathrm{new}}(x_S) = L_S \Phi_S^{\mathrm{old}}(x_S)
\end{equation*}
For any query $a$, we have $\langle a, L_S \Phi_S^{\mathrm{old}}(x_S) \rangle = \langle L_S^T a, \Phi_S^{\mathrm{old}}(x_S) \rangle$. Because $\mathcal{C}_S$ preserves all linear tests on $\Phi_S^{\mathrm{old}}$, it preserves all linear tests on $\Phi_S^{\mathrm{new}}$ without introducing additional error or domain expansion. When a sum layer follows a product layer at the same scope, the algorithm first sparsifies the product-stage features and then applies this linear map only to the surviving rows; no additional sparsification is performed between these two stages.

\paragraph{Product Gates.} Let $S_g = S_h \sqcup S_\ell$ be a product region formed by disjoint child scopes $S_h$ and $S_\ell$. All product gates of either circuit at this region use the split prescribed by the common v-tree and are processed together. Let $\Psi_{S_g}(x_{S_g})$ concatenate the values of these product gates in $C_P$ and $C_Q$ before any subsequent sum gates of scope $S_g$ are evaluated. Decomposability makes $\Psi_{S_g}$ a bilinear function of the two child feature vectors, and the complete feature vector $\Phi_{S_g}$ is obtained from $\Psi_{S_g}$ by a fixed block-diagonal linear map that also retains any required product-gate coordinates. Given reduced child domains $\mathcal{C}_{S_h}$ and $\mathcal{C}_{S_\ell}$, we form the candidate Cartesian product domain $\mathcal{U}_{S_g} = \mathcal{C}_{S_h} \times \mathcal{C}_{S_\ell}$ containing $N = m_{S_h} m_{S_\ell}$ pairs $z_{ij} = (u_i, v_j)$ with inherited weights $\mu_{ij} = \lambda_i \rho_j$. For every candidate pair, we evaluate $\Psi_{S_g}(z_{ij})$ and apply Theorem \ref{thm:lewis_weights} with parameters $\delta$ and $\eta'$. This produces a reweighted subset of size $m_{S_g} = O\left( \frac{W \log(2W)}{\delta^2} \log\frac{1}{\eta'} \right)$. Applying the same-scope linear map to the surviving rows yields the coreset $\mathcal C_{S_g}$.

\subsection{Approximation Guarantee and Complexity}

\begin{theorem}[Probabilistic Circuit TV Approximation]
\label{thm:pc_fpras}
Given $\epsilon, \eta \in (0, 1)$ and smooth, decomposable probabilistic circuits $C_P$ and $C_Q$ computing $P$ and $Q$, respectively, and structured with respect to a common v-tree, let $W$, $q$, and $L\geq1$ be as defined above. The bottom-up coreset propagation outputs the estimator $\widetilde{D}$ defined above, satisfying $(1-\epsilon)d_{\mathrm{TV}}(P,Q) \le \widetilde{D} \le (1+\epsilon)d_{\mathrm{TV}}(P,Q)$ with probability at least $1-\eta$. Let $M$ denote the maximum number of rows retained in any
internal-region coreset, and
let $|C|:=|C_P|+|C_Q|$, where each circuit size counts its gates and wires. The arithmetic running time is
\[
\widetilde{O}\!\left((q+M)|C|+L\left[W(q+M)^2+W^\omega\right]\right),
\quad\text{where }
M=O\!\left(
\frac{WL^2\log(2W)}{\epsilon^2}
\log\frac{L}{\eta}
\right).
\]
\end{theorem}

\begin{lemma}[Coreset Propagation Invariant]
\label{lem:pc_invariant}
For a leaf or active product region $S$ for which $\mathcal C_S$ is constructed, let $\ell(S)$ be the number of active product regions in the subtree rooted at $S$. Conditioned on $\textsc{Sparsify}$ succeeding at all corresponding product-region sparsification steps, it holds simultaneously for all $a \in \mathbb{R}^{d_S}$ that:
\begin{equation*}
    (1 - \delta)^{\ell(S)} \sum_{x_S \in \Omega_S} |\langle a, \Phi_S(x_S) \rangle| \le \sum_{(z,\mu) \in \mathcal{C}_S} \mu |\langle a, \Phi_S(z) \rangle| \le (1 + \delta)^{\ell(S)} \sum_{x_S \in \Omega_S} |\langle a, \Phi_S(x_S) \rangle|
\end{equation*}
\end{lemma}
\begin{proof}
The base case holds trivially at the leaves. Sum gates apply fixed linear maps and introduce no error. At a product region $S_g = S_h \sqcup S_\ell$, write the complete same-scope computation as $\Phi_{S_g}=L_{S_g}\Psi_{S_g}$, where $\Psi_{S_g}$ is the product-stage feature vector and $L_{S_g}$ is linear. For every query $a$, there is a matrix $M_a$ such that $\langle a,\Phi_{S_g}(x_h,x_\ell)\rangle=\Phi_{S_h}(x_h)^T M_a\Phi_{S_\ell}(x_\ell)$. Fixing $x_\ell$, this is a linear query on $\Phi_{S_h}$, so the left child invariant preserves the sum over $x_h$ up to $(1 \pm \delta)^{\ell(S_h)}$. Fixing the surviving left assignments, it is a linear query on $\Phi_{S_\ell}$, so the right child invariant preserves the sum over $x_\ell$ up to $(1 \pm \delta)^{\ell(S_\ell)}$. Lewis sparsification of the product-stage candidate set contributes one multiplicative $(1 \pm \delta)$ factor, after which applying $L_{S_g}$ to the surviving rows introduces no further error. Since $\ell(S_g) = \ell(S_h) + \ell(S_\ell) + 1$, the invariant holds.
\end{proof}

\begin{proof}[Proof of Theorem \ref{thm:pc_fpras}]
By the union bound over the $L$ product-region sparsification steps, all sparsification steps succeed with probability at least $1 - L\eta' = 1 - \eta$. At the root, the exact sum in Lemma \ref{lem:pc_invariant} equals $2d_{\mathrm{TV}}(P,Q)$, while the coreset sum equals $2\widetilde D$ by definition. Hence the lemma gives $(1-\delta)^L 2d_{\mathrm{TV}}(P,Q) \le 2\widetilde D \le (1+\delta)^L 2d_{\mathrm{TV}}(P,Q)$. Dividing by $2$ and substituting $\delta = \frac{\epsilon}{3L}$ yields the $(1 \pm \epsilon)$ relative bounds.

For complexity, let $M=O\left(\frac{W\log(2W)}{\delta^2}\log\frac{1}{\eta'}\right)=O\left(\frac{WL^2\log(2W)}{\epsilon^2}\log\frac{L}{\eta}\right)$ be the maximum size of an internal-region coreset. A child domain has at most $q$ rows when it is a leaf and at most $M$ rows otherwise, so every binary product region has $|\mathcal{U}_S|\leq(q+M)^2$ candidate rows. The product-stage feature matrix has $O(W)$ columns, and Theorem \ref{thm:lewis_weights} therefore computes its Lewis weights in $\widetilde{O}\left((q+M)^2W+W^\omega\right)$ arithmetic operations (\cite[Theorem 5.3.1]{lee2016faster}; \cite[Lemma 2.5]{jambulapati2022improved}). Across the two circuits, evaluating the exact leaf tables costs $O(q|C|)$ operations, while applying all subsequent gates and wires only to the surviving rows costs $O(M|C|)$. Adding these costs over the $L$ product regions gives $\widetilde{O}\left((q+M)|C|+L[W(q+M)^2+W^\omega]\right)$, as claimed.
\end{proof}

\section{Application to Weighted Tree Automata}
We now apply the result of Section~\ref{sec:pc} to a fixed-tree
formulation of weighted tree automata. We use the bottom-up tensor
representation \cite{FulopVogler09,RabusseauBalleCohen16}.
The tree is binary, its leaf labelings form the sample space, and
all parameters are nonnegative. We allow the state dimensions and
transition tensors to depend on the node. Each coordinate of a
bilinear transition is a weighted sum of products of child-state
coordinates. We can therefore express the two models using sum
and product gates that respect the common tree.

Let $\mathcal{T}=(V,E)$ be a rooted binary tree with root $r$ and leaf set $L$. Its leaves are identified with observed variables $X_1,\ldots,X_{|L|}$, where $X_i\in\Omega_i$ and every $\Omega_i$ is finite. For a node $u\in V$, let $S_u\subseteq[|L|]$ be the set of leaf indices in the subtree rooted at $u$, and define $\Omega_{S_u}:=\prod_{i\in S_u}\Omega_i$. If $u$ is internal, denote its left and right children by $u_{\mathrm L}$ and $u_{\mathrm R}$; then $S_u=S_{u_{\mathrm L}}\sqcup S_{u_{\mathrm R}}$. The common sample space is $\Omega:=\Omega_{S_r}=\prod_{i=1}^{|L|}\Omega_i$. Thus a sample $x\in\Omega$ is a labeling of the leaves of $\mathcal{T}$.

Fix a model $R\in\{P,Q\}$. Every node $u\in V$ carries a positive interface dimension $d_u^R$ and exactly one gate. Evaluating the gates bottom-up produces, at every node $u$, a message $h_u^R:\Omega_{S_u}\to\mathbb R_{\geq0}^{d_u^R}$, which is the vector that the subtree rooted at $u$ exposes on its output interface as a function of the leaf labels below $u$. The gates fall into three types according to the position of $u$ in $\mathcal{T}$: a leaf reads the observed symbol, an internal node combines the two messages of its children bilinearly, and the root does the same but outputs a scalar.

\paragraph{Leaf gates.}
If a leaf $\ell$ corresponds to variable $X_i$, its gate is a nonnegative lookup table $h_\ell^R:\Omega_i\to\mathbb R_{\geq0}^{d_\ell^R}$, which assigns a $d_\ell^R$-dimensional feature vector to every symbol $x_i\in\Omega_i$.

\paragraph{Internal gates.}
Every internal node $u$ with children $u_{\mathrm L},u_{\mathrm R}$ has a nonnegative bilinear gate $\mathcal B_u^R:\mathbb R^{d_{u_{\mathrm L}}^R}\times\mathbb R^{d_{u_{\mathrm R}}^R}\to\mathbb R^{d_u^R}$. Such a gate can be interpreted as a third-order tensor $T_u^R\in\mathbb R_{\geq0}^{d_u^R\times d_{u_{\mathrm L}}^R\times d_{u_{\mathrm R}}^R}$ with the $i$-th coordinate
\begin{equation*}
\bigl[\mathcal B_u^R(p,q)\bigr]_i = \sum_{j=1}^{d_{u_{\mathrm L}}^R}\sum_{k=1}^{d_{u_{\mathrm R}}^R} T_{u,i,j,k}^R\,p_jq_k .
\end{equation*}
The gate is applied to the two child messages, so that
$$
h_u^R(x_{S_u}) := \mathcal B_u^R\!\left(h_{u_{\mathrm L}}^R(x_{S_{u_{\mathrm L}}}), h_{u_{\mathrm R}}^R(x_{S_{u_{\mathrm R}}})\right),
$$
where $x_{S_u}=(x_{S_{u_{\mathrm L}}},x_{S_{u_{\mathrm R}}})$ since $S_u=S_{u_{\mathrm L}}\sqcup S_{u_{\mathrm R}}$.

\paragraph{Root output and normalization.}
Since $d_r^R=1$, the root message is scalar. For a leaf labeling $x\in\Omega$, its unnormalized weight is $f_R(x):=h_r^R(x)$. Define the normalization factor $Z_R:=\sum_{x\in\Omega}f_R(x)$ and assume $Z_R>0$; the distribution $R\in\{P,Q\}$ is then given by
$$
R(x) = \frac{f_R(x)}{Z_R}.
$$

The normalizing constants can be computed exactly by a bottom-up dynamic program, without enumerating the full sample space $\Omega$. For each model $R\in\{P,Q\}$, define $m_\ell^R:=\sum_{x_i\in\Omega_i}h_\ell^R(x_i)$ at a leaf $\ell$ corresponding to $X_i$, and $m_u^R:=\mathcal B_u^R(m_{u_{\mathrm L}}^R,m_{u_{\mathrm R}}^R)$ at an internal node $u$. By bilinearity, a bottom-up induction gives $m_u^R=\sum_{x_{S_u}\in\Omega_{S_u}}h_u^R(x_{S_u})$ for every node $u$; at the root this reads $m_r^R=\sum_{x\in\Omega}f_R(x)=Z_R$.

We now run the two models in parallel. For every node $u$, define the joint interface dimension $d_u:=d_u^P+d_u^Q$ and the joint feature
\begin{equation*}
\Phi_u(x_{S_u}) := \begin{bmatrix} h_u^P(x_{S_u})\\ h_u^Q(x_{S_u}) \end{bmatrix} \in\mathbb R_{\geq0}^{d_u}.
\end{equation*}
At every internal node $u$, define the joint gate $\mathcal B_u:\mathbb R^{d_{u_{\mathrm L}}}\times\mathbb R^{d_{u_{\mathrm R}}}\to\mathbb R^{d_u}$ by
\begin{equation*}
\mathcal B_u\left(\begin{bmatrix}p^P\\p^Q\end{bmatrix},\begin{bmatrix}q^P\\q^Q\end{bmatrix}\right) := \begin{bmatrix} \mathcal B_u^P(p^P,q^P)\\ \mathcal B_u^Q(p^Q,q^Q) \end{bmatrix}.
\end{equation*}
This joint gate is bilinear because it is the direct sum of the two bilinear gates $\mathcal B_u^P$ and $\mathcal B_u^Q$.

By construction, the joint features satisfy
\begin{equation*}
\Phi_u(x_{S_u}) = \mathcal B_u\left(\Phi_{u_{\mathrm L}}(x_{S_{u_{\mathrm L}}}), \Phi_{u_{\mathrm R}}(x_{S_{u_{\mathrm R}}})\right).
\end{equation*}

Since $d_r^P=d_r^Q=1$, the joint root feature is two-dimensional, $\Phi_r(x)=(f_P(x),f_Q(x))^{\mathsf T}$. Define the fixed root query $a_{\mathrm{TV}}:=(Z_P^{-1},-Z_Q^{-1})^{\mathsf T}$. For every leaf labeling $x\in\Omega$, we then have
\begin{equation}
\label{eq:wta-pointwise-difference}
\left\langle a_{\mathrm{TV}},\Phi_r(x)\right\rangle = \frac{f_P(x)}{Z_P}-\frac{f_Q(x)}{Z_Q} = P(x)-Q(x).
\end{equation}
Let $\mathcal D_r:=\{(x,1,\Phi_r(x)):x\in\Omega\}$ be the exact weighted feature multiset at the root. By \eqref{eq:evaluation} and \eqref{eq:wta-pointwise-difference}, $E(\mathcal D_r,a_{\mathrm{TV}})=\sum_{x\in\Omega}|P(x)-Q(x)|=2d_{\mathrm{TV}}(P,Q)$, that is,
\begin{equation*}
d_{\mathrm{TV}}(P,Q) = \frac12 E(\mathcal D_r,a_{\mathrm{TV}}).
\end{equation*}

To apply Section~\ref{sec:pc}, we expand each bilinear gate into scalar product gates followed by scalar sum gates.

\begin{theorem}[FPRAS for TV-distance between weighted tree automata]\label{thm:wta_fpras}
Given $\varepsilon,\eta\in(0,1)$, two distributions $P$ and $Q$ induced by two nonnegative weighted tree automata on the same tree $\mathcal{T}$, there exists an FPRAS that outputs a real number $\widehat d$ such that
$(1-\varepsilon)d_{\mathrm{TV}}(P,Q)\leq\widehat d\leq(1+\varepsilon)d_{\mathrm{TV}}(P,Q)$, with probability at least $1-\eta$ in time
$$
\widetilde{O}\!\left(|V|d^2\left(q+d+\frac{d^2|V|^2}{\varepsilon^2}\log\frac{|V|}{\eta}\right)^2+|V|d^{2\omega}\right),
$$
where $q:=\max_{i\in[|L|]}|\Omega_i|$ and $d:=\max_{u\in V} d_u$ are the maximum leaf-domain size and joint node-interface dimension, respectively.
\end{theorem}

\begin{proof}
At every node $u$, we pad each model's message vector with zero coordinates to the common global dimension $d$. The corresponding leaf tables and transition tensors are padded with zeros. This padding does not change either induced distribution. For notational simplicity, we henceforth reuse $h_u^R$ and $T_u^R$ to denote the corresponding zero-padded messages and tensors. Since $Z_R$ can be computed exactly by the preceding bottom-up dynamic program, we append a unary root sum gate with weight $Z_R^{-1}$, whose output is $R(x)=Z_R^{-1}f_R(x)$. In the expanded root feature space, the two-dimensional query $a_{\mathrm{TV}}$ defined above is extended by zero coefficients on all other root-scope gate coordinates.

We next represent every bilinear gate $\mathcal B_u^R$ exactly as a layer of scalar product gates followed by a layer of scalar sum gates. For every node $u\in V$ and every $i\in[d]$, let $g_{u,i}$ denote the gate corresponding to the $i$-th coordinate of the padded message at node $u$. At a leaf $\ell$, $g_{\ell,i}$ is a leaf gate satisfying
$$
R_{g_{\ell,i}}(x_{S_\ell}) = \bigl[h_\ell^R(x_{S_\ell})\bigr]_i
$$
under each parameterization $R\in\{P,Q\}$.

Now let $u$ be an internal node. Recall that its bilinear gate is given coordinatewise by
$$
\bigl[\mathcal B_u^R(p,q)\bigr]_i = \sum_{j=1}^{d}\sum_{k=1}^{d} T_{u,i,j,k}^R\,p_jq_k, \qquad i\in[d].
$$
For every $j\in[d]$ and $k\in[d]$, introduce a product gate $p_{u,j,k}$ whose children are $g_{u_{\mathrm L},j}$ and $g_{u_{\mathrm R},k}$. Under parameterization $R$, it computes
$$
R_{p_{u,j,k}}(x_{S_u}) = R_{g_{u_{\mathrm L},j}}(x_{S_{u_{\mathrm L}}}) R_{g_{u_{\mathrm R},k}}(x_{S_{u_{\mathrm R}}}).
$$

For every $i\in[d]$, define $g_{u,i}$ to be a sum gate whose children are all product gates $p_{u,j,k}$, with edge weight $T_{u,i,j,k}^R$ under parameterization $R$. Thus,
$$
R_{g_{u,i}}(x_{S_u}) = \sum_{j=1}^{d}\sum_{k=1}^{d} T_{u,i,j,k}^R\, R_{p_{u,j,k}}(x_{S_u}).
$$
Hence this two-layer subcircuit computes the bilinear gate:
$$
\begin{bmatrix} R_{g_{u,1}}(x_{S_u})\\ \vdots\\ R_{g_{u,d}}(x_{S_u}) \end{bmatrix} = \mathcal B_u^R\!\left(h_{u_{\mathrm L}}^R(x_{S_{u_{\mathrm L}}}), h_{u_{\mathrm R}}^R(x_{S_{u_{\mathrm R}}})\right) = h_u^R(x_{S_u}).
$$

Each product gate $p_{u,j,k}$ is decomposable because its two children have the disjoint scopes $S_{u_{\mathrm L}}$ and $S_{u_{\mathrm R}}$. Each sum gate $g_{u,i}$ is smooth because all of its children have scope $S_u=S_{u_{\mathrm L}}\sqcup S_{u_{\mathrm R}}$. All product gates respect $\mathcal T$, so the two expanded circuits are structured with respect to the same v-tree. Therefore, Theorem~\ref{thm:pc_fpras} gives the stated approximation guarantee.

For the running time, let $I:=|V\setminus L|$ be the number of internal nodes. If $I=0$, the TV-distance can be computed exactly by enumerating the single leaf domain, so assume $I\geq1$. Each expanded circuit has at most $d^2$ product gates and $d$ sum gates per internal node; hence $W\leq d^2+d+1=O(d^2)$, the theorem's product-region parameter is $I\leq|V|$, and, counting gates and wires, $|C|=|C_P|+|C_Q|=O(|L|d+Id^3)$. The preliminary normalization dynamic program costs $O(|L|qd+Id^3)$ and is subsumed by the claimed bound. Substituting these quantities into Theorem~\ref{thm:pc_fpras} gives the stated running time directly.
\end{proof}

\section*{Acknowledgements}
The authors thank Barath Ashok for his involvement in the early stages of this work. They also thank Weiming Feng for putting them in touch and helping bring about this collaboration. 

The main result of this paper was obtained independently, around the same time, by Yucheng Fu and by the other authors. 
\ignore{After completing his work on zonotope compression~\cite{Fu26}, Fu obtained the result following a suggestion from GPT-5.6 Pro. When Arnab Bhattacharyya informed Weiming Feng of the other authors' result, Feng pointed out that Fu had also obtained it and introduced the authors by email. The authors then agreed to present their results jointly in this paper.}
The authors used Gemini in the early stages of this work, primarily to check arguments they had proposed. In the later stages, they used ChatGPT (GPT-5.6 Sol and GPT-6 Astra) to draft portions of the manuscript and to assist with writing. The Lean4 formalization was accompanied with the aid of Tex2Lean~\footnote{The tool is available at \url{https://marketplace.visualstudio.com/items?itemName=kuldeepmeel.tex2lean4}}, which in turn uses Claude and Codex.

\bibliography{refs}
\bibliographystyle{alpha}

\end{document}